\documentclass[letterpaper, 12pt]{article}
\usepackage{amsmath,amssymb,amsthm,setspace,latexsym}
\usepackage{epsfig,psfrag,graphicx,epstopdf}
\usepackage[usenames,dvipsnames]{color}
\usepackage[top=2cm, bottom=2cm, left=2cm, right=2cm]{geometry}
\usepackage{authblk}
\newcommand{\url}{}

\def\cF{{\mathcal F}}

\def\cF{{ F}}
\def\bE{{ E}}
\def\bq{{ q}}
\def\bc{{ c}}
\def\bv{{ v}}
\def\bw{{ w}}
\def\bF{{\mathbf F}}

\def\bj{{ j}}
\def\bV{{\mathbf V}}
\def\bW{{\mathbf W}}
\def\bQ{{\mathbf Q}}

\def\by{{y}}
\def\bY{{\mathbf Y}}
\def\cK{{\mathcal K}}

\def\cS{{\mathcal S}}

\newtheorem{thm}{THEOREM}[section]
\newtheorem{lem}{LEMMA}[section]
\newtheorem{cor}{COROLLARY}[section]
\theoremstyle{definition}

\usepackage{color}

\begin{document}
 
\title{Entropic Chaoticity for the Steady State of a Current Carrying System.}

\author[*]{F. Bonetto}
\author[*]{M. Loss}
\affil[*]{School of Mathematics, Georgia Institute of Technology, Atlanta, GA}

\renewcommand\Authands{ and }

\maketitle

\begin{abstract}
The steady state for a system of $N$ particles under the influence of an 
external field and a Gaussian thermostat and colliding with random ``virtual'' 
scatterers can be obtained explicitly in the limit of small fields. We show that 
the sequence of steady state distributions, as $N$ varies, is a chaotic sequence 
in the sense that the $k$ particle marginal, in the limit of large $N$, is the 
$k$-fold tensor product of the 1 particle marginal. We also show that the 
chaoticity properties holds in the stronger form of entropic chaoticity.
\end{abstract}

\section{Introduction}
Since its introduction by Kac \cite{Kac} in 1956, the notion of a chaotic 
sequence has become an important concept in studying many body systems. Chaotic 
sequences and propagation of chaos are the principal tools for constructing 
effective equations for many body problems. The aim of this article is to give 
yet another example of the interplay between chaotic sequences and effective 
equations.

In \cite{BDLR} the authors consider a system consisting of $N$ particles moving 
in a 2-dimensional torus and colliding with convex scatterers that form a 
dispersing billiard. The particles are subject to an external electric field 
$\bE$ and a Gaussian thermostat that keeps the kinetic energy fixed. The 
equation of motion between collisions are given by
\begin{equation}\label{eq2}
\left\{
\begin{array}{l}
   \dot\bq_i = \bv_i  \qquad\qquad\qquad\qquad i=1,\ldots,N \\
   \dot\bv_i = \bF_i = \bE-\frac{\displaystyle \bE\cdot
\bj}{\displaystyle U}\,\bv_i + \cF_i  \ ,\end{array}\ 
\right.
\end{equation}
where
\begin{equation}  \label{eq3}
    \bj(\bV)=\sum_{i=1}^N \bv_i,\qquad U(\bV) = \sum_{i=1}^N |\bv_i|^2 \ ,
\end{equation}
and $\cF_i$ is the force exerted on the $i$-th particle by collisions with the 
fixed scatterers. We use the notation $\bV = (v_1, v_2, \dots, v_N)$ where $v_i 
\in \mathbb R^2, i=1, \dots, N$. Very little is known about billiards with more 
than one particle. In particular there is no existence theorem for the SRB 
measure of this system.

The authors introduced in  \cite{BDLR}  a stochastic version of the above model 
in which the deterministic collisions are replaced by Poisson distributed 
collision processes. More precisely, in the time interval between $t$ and 
$t+dt$, the $i$-th particle has a probability $|\bv_i|^\alpha dt$ of suffering a 
collision. When a collision happens, the velocity of the particle is randomly 
updated, i.e. if the particle's velocity direction before the collision is 
$\omega=\bv/|\bv|$,  after the collision it will be distributed as 
$K(\omega'\cdot \omega)d\omega'$. The details of the collision kernel $K$ are 
largely irrelevant. For what follows it will be enough to assume that $K(x)>0$ 
for $x$ in a open set ${\cal U}\in[-1,1]$. We note that this stochastic process 
makes sense for any dimension $d$. We shall use, as mentioned before, the notation 
$\bV=(v_1, \dots, v_N)$, $\bQ=(q_1, \dots, q_N)$, with $v_i \in \mathbb R^d$, 
and $q_i \in  \mathbb{T}^d$ , $i=1, \dots , N$.

Let $W(\bQ,\bV,t;\bE)$ be the probability density of finding the particles at 
positions $\bQ$ with velocities $\bV$ at time $t$. The time evolution of 
$W(\bQ,\bV,t;\bE)$ is given by the master equation 
\begin{eqnarray}  \label{dW}
 \frac{\partial W(\bQ,\bV,t;\bE)}{\partial t} = 
&-&\sum_{i=1}^N\bv_i\cdot\nabla_{q_i} W(\bQ,\bV,t;\bE)
-\sum_{i=1}^N
 \nabla_{\bv_i}\left[\left(\bE-\frac{\bE\cdot\bj(\bV)}{U(\bV)}
\bv_i\right)W(\bQ,\bV,t;\bE)
\right]\nonumber\\
  &+&\sum_{i=1}^N|\bv_i|^\alpha\int_{\cS^{d-1}(1)} K(\omega_i \cdot \omega'_i)
\Bigl[W(\bQ,\bV_i',t;\bE)-W(\bQ,\bV,t;\bE)\Bigr]\,d\sigma^{d-1}
(\omega')\nonumber\\
\end{eqnarray}
where $\cS^{m}(R)$ is the $m$-dimensional sphere with radius $R$, and $d 
\sigma^{m}(\cdot)$ is the uniform, surfaces measure on $\cS^{m}(R)$. Further, if 
$\bV = (\bv_1,\ldots,\bv_i,\ldots,\bv_N)$ then  $\bV_i' = 
(\bv_1,\ldots,\bv_i',\ldots,\bv_N)$, with $\bv_i'=|\bv_i|\omega_i'$ if 
$\bv_i=|\bv_i|\omega_i$. Note that the variable $\bQ$ is not part of the 
dynamics, i.e., if the initial condition $W(\bQ,\bV,0)$ is independent of $\bQ$ 
so will be $W(\bQ,\bV,t; E)$. Moreover, if $W(\bQ,\bV,0)$ is concentrated on the 
surface of a given energy $U_0$, that is if 
\[
 W(\bQ,\bV,0)=\delta(U(\bV)-U_0)F(\bQ,\bV,0) \ ,
\]
then so will be the solution of \eqref{dW}:
\[
 W(\bQ,\bV,t;\bE)=\delta(U(\bV)-U_0)F(\bQ,\bV,t;\bE).
\]
Finally, if $F(\bQ,\bV,0)$ is a symmetric function so is $F(\bQ,\bV,t; \bE)$. 
Thus, from now on we will only consider symmetric, spatially homogeneous 
solutions concentrated on the surface of energy $U_0=N$, that is, on the $dN-1$ 
dimensional sphere $\cS^{dN-1}(\sqrt{N})$ of radius $\sqrt{N}$. In particular, 
this means that $F(\bQ,\bV,t;\bE)=F(\bV,t;\bE)$ will not depend on the positions 
$\bQ$.

Recall that the $k$-particle marginal $l^{(k)}_N(\bv_1, \dots, \bv_k)$  of a
distribution $L_N(\bV)$ on $\cS^{dN-1}(\sqrt{N})$ is defined by the equations
\[
\int_{\cS^{dN-1}(\sqrt{N})} \varphi(\bv_1, \dots, \bv_k) L_N(\bV) d
\sigma^{dN-1}(\bV) =
\int_{\mathbb{R}^{dk}} \varphi(\bv_1, \dots, \bv_k)
l^{(k)}_N(\bv_1, \dots, \bv_k) d \bv_1 \cdots d \bv_k  \ ,
\]
where $\varphi(\bv_1, \dots, \bv_k)$ ranges over the set of bounded continuous 
function on $\mathbb{R}^{dk}$. Simple computations show that 
\begin{equation} \label{kmarginal} 
l^{(k)}_N(\bV_k) = \sqrt{\frac{N}{N - |\bV_k|^2}} 
\int_{\cS^{d(N-k)-1}\left(\sqrt{N - |\bV_k|^2}\right)} L_N(\bV_k, \bV^k) d 
\sigma^{d(N-k)-1}(\bV^k) \ ,
\end{equation}
where $\bV_k = (\bv_{1} , \dots, \bv_k)$ and $\bV^k = (\bv_{k+1} , \dots, 
\bv_N)$. A sequence of densities $\{L_N\}_{N=1}^\infty $ forms a {\bf chaotic 
sequence with marginal $l$} if for any bounded continuous function $\varphi$
\begin{equation}\label{weak}
\lim_{N \to \infty} \int_{\cS^{dN-1}(\sqrt{N})} \varphi(\bV_k)
L_N(\bV) d \sigma^{dN-1}(\bV) =
\int_{\mathbb{R}^{dk}} \varphi(\bV_k)
\prod_{j=1}^k l(\bv_j)  d \bv_1 \cdots d \bv_k \ .
\end{equation}

It was shown in \cite{BCELM} that for finite time $t$ the master equation 
\eqref{dW} propagates chaos, i.e.,  the solution of the master equation 
\eqref{dW} forms a chaotic sequence  if the initial condition does. More 
precisely,  if for any bounded continuous function $\varphi(\bV_k)$ the initial 
condition $F_N(V,0)$ for the master equation \eqref{dW} satisfies
\[
\lim_{N \to \infty} \int_{\cS^{dN-1} } \varphi(\bV_k) F_N(\bV, 0) d\sigma(\bV)
=\int_{\mathbb{R}^{dk}}  \varphi(\bV_k) \prod_{j=1}^k f  (\bv_j, 0) d\bv_1
\cdots d\bv_k \ ,
\]
then
\[
\lim_{N \to \infty} \int_{\cS^{dN-1} } \varphi(\bV_k) F_N(\bV, t;
\bE) d\sigma(\bV) =
\int_{\mathbb{R}^{dk}}  \varphi(\bV_k) \prod_{j=1}^k f  (\bv_j, t;
\bE) d\bv_1 \cdots d\bv_k
\]
where
\[
  f (\bv ,t;\bE)=\lim_{N\to\infty}  f^{(1)}_N(\bv ,t;\bE) 
\]
satisfies the Boltzmann equation
\begin{equation}\label{BE} 
\frac{f(\bv,t;\bE)}{dt}+\nabla_\bv\left[\left(\bE-\frac{
\bE\cdot\hat\bj(t,E)}{u} \bv\right)f
\right]=|\bv|^\alpha\int_{\cS^{d-1}(1)} K(\omega \cdot \omega')
  \bigl[f(\bv',t;\bE)-f(\bv,t;\bE)\bigr]\,d\sigma^{d-1}(\omega') \ .
\end{equation}
Here,  $\hat \bj(t,E)$ and $u$ are given by the self consistent condition
\[
 \hat \bj(t,E)=\int \bv f(\bv,t;\bE) d\bv 
\qquad {\rm and} \qquad 
 u=\int |\bv|^2 f(\bv,t;\bE) d\bv.
\]
 It is easy to see that $u$ is independent of time and, since we have chosen
$U_0=N$, $u=1$. The initial condition is given by $f(\bv) = \lim_{N \to \infty} 
f^{(1)}_N(\bv,0)$.

Concerning the steady states, the situation is far from clear. In \cite{BDLR} 
and \cite{BCKLs} it was shown that a steady state $F_{\rm ss} (\bV;\bE)$ exists 
for the master equation \eqref{dW} provided that  $\bE \not= 0$. If $\bE = 0$ 
any density $F(\bV)$ that depends only on the magnitude of the velocities 
furnishes a stationary state. It is, however, an open question whether $F_{\rm 
ss} (\bV;\bE)$ tends to a limiting distribution as $\bE \to 0$. Interestingly, 
assuming that a limiting distribution {\it exists}, it can be computed exactly 
and it is given by
\begin{equation}\label{exN}
F_{\rm
ss}(\bV;0)=\delta(U(\bV)-N)\frac{1}{\widetilde{Z}_N}\frac{1}{
\left(\sum_{i=1}^N
|\bv_i|^{2+\alpha}\right)^{\frac{dN-1}{2+\alpha}} }:=\delta(U(\bV)-N)F_N(\bV) \ ,
\end{equation}
where $\widetilde{Z}_N$ is the normalization constant
\begin{equation} \label{normalization}
\widetilde{Z}_N= \int_{\cS^{dN-1}(\sqrt{N})} \frac{ d\sigma^{Nd-1} (\bV)}{
\left(\sum_{i=1}^N
|\bv_i|^{2+\alpha}\right)^{\frac{dN-1}{2+\alpha}} }\ .
\end{equation}
For details, the reader should consult \cite{BDLR} and \cite{BCKLs}. Thus, the 
electric field `selects' the right steady state as it tends to zero.

A similar problem exists on the level of the Boltzmann equation. Again it is
possible to show that the steady state $f_{\rm ss}(\bv;\bE)$ for the Boltzmann
equation \eqref{BE} exists and is unique if $E\not=0$. This clearly implies the
existence of a steady state current $\hat \bj_{\rm ss}(E)$. In \cite{BL} it
was shown that, assuming that $f(v)=\lim_{E\to 0} f_{\rm ss}(v,E)$ exists and
that $\hat \bj_{\rm ss}(E)=O(E)$, one has
\begin{equation}\label{exin}
 f(\bv)=\frac{\mu^{\frac{d}{2}}}{c}e^{-\left(\sqrt{\mu} 
|\bv|\right)^{2+\alpha}} \ ,
\end{equation}
where $c$ and $\mu$ are uniquely determined by the normalization  of $f$ and the
condition $u=1$. One easily computes
\begin{equation} \label{constants}
 c:= \int_{\mathbb R^d} e^{- |v|^{2+\alpha}} d v = \frac{2 
\pi^{\frac{d}{2}}}{\Gamma\left(\frac{d}{2}\right)}\frac{\Gamma\left(\frac{d}{
2+\alpha
} \right) }{(2+\alpha)}
\qquad  {\rm and } \qquad \mu:=\frac{1}{c} \int_{\mathbb R^d} |v|^2 
e^{-|v|^{2+\alpha}} dv = 
\frac{\Gamma\left(\frac{d+2}{2+\alpha}\right)}{\Gamma\left(\frac{d}{ 
2+\alpha} \right)}\ ,
\end{equation}
which for $\alpha=1$ and $d=2$ yield
\[
\mu=\frac{\Gamma\left(\frac{4}{3}\right)}{\Gamma\left(\frac{2}{3}\right)}
\approx 0.65948\qquad {\rm and} \qquad
c=\frac{2\pi}{3}\Gamma\left(\frac{2}{3}\right)
 \approx 2.83605 \ .
\]
For details the reader may consult \cite{BL} where the existence of the small
field limit of the steady state distribution is proved for $d=1$.

It is now natural to ask whether the distribution $F_N$ defined in \eqref{exN} 
is chaotic with marginal $f$ given by  \eqref{exin}. This cannot be deduced from 
the previous results on propagation of chaos since those results do not hold 
uniformly in time. A more serious impediment is the fact that the small field 
limits of the steady states are not known to exist. As explained before, the 
limit as $\bE \to 0$ seems to select a steady state for the master equation as 
well as for its Boltzmann version. It is far from clear that the selection 
mechanism is such as to preserve chaoticity. In this note we prove that, 
nevertheless, the distribution defined in \eqref{exN} is chaotic with marginal 
\eqref{exin}.

\medskip
\begin{thm} \label{chaoticss}
 Let $f^{(1)}_N(\bv)$ be the one particle marginal of $F_N(\bV)$ defined by
 \begin{equation}\label{NN}
  f^{(1)}_N(\bv_1)=\frac{\sqrt{N}}{\sqrt{N-
 |\bv_1|^2}} 
\frac{1}{\widetilde{Z}_N}\int_{\cS^{d(N-1)-1}\left(\sqrt{N-|\bv_1|^2}\right)}
F_N(\bV)d\sigma(\bV^1)
\end{equation}
and set
\begin{equation}
f(\bv)=\frac{\mu^{\frac{d}{2}}}{c}e^{-\left(\sqrt{\mu} |\bv|\right)^{2+\alpha}}
\end{equation} 
with the constants given by \eqref{constants}.  Then for any bounded continuous 
function $\varphi(\bv)$ 
\begin{equation}\label{prove}
 \lim_{N\to\infty} \int_{\mathbb{R}^d} \varphi(\bv) f_N^{(1)}(\bv) d \bv 
=\int_{\mathbb{R}^d} \varphi(\bv) f(\bv) d \bv 
\end{equation}
and for every $k$,
the $k$ particle marginal $f^{(k)}_N(\bv_1,\ldots,\bv_k)$ of $F_N(\bV)$  
satisfies
\begin{equation}\label{prove1}
 \lim_{N\to\infty}\int_{\mathbb{R}^{kd}} \varphi(\bv_1,\ldots,\bv_k)  f^{(k)}_N(\bv_1,\ldots,\bv_k) d \bv_1 \cdots d \bv_k =
 \int_{\mathbb{R}^{kd}} \varphi(\bv_1,\ldots,\bv_k)  \prod_{i=1}^k f(\bv_k)  d \bv_1 \cdots d \bv_k
\end{equation}
where, again, $\varphi$ is any bounded continuous function on $\mathbb{R}^{kd}$. 
Thus $F_N(\bV)$ from a chaotic sequence with marginal $f$.
\end{thm}
\medskip

\section{Proof of Theorem \ref{chaoticss}.}

The following elementary lemma sets the stage for the proof. It will be
expressed in terms of the probability distribution 
\[
g(\bw) := \frac{e^{-|\bw|^{2+\alpha}}}{\int_{\mathbb{R}^d} e^{- |\bw|^{2+\alpha}} d\bw } \ .
\]
In addition to the constants  $c$ and $\mu$ given by  \eqref{constants} we need
\begin{equation} \label{sigma}
\sigma^2:=\int_{\mathbb{R}^d} (|\bw|^2-\mu^2)^2 
g(\bw)d\bw= 
\frac{\Gamma\left(\frac{d+4}{2+\alpha}\right)}
{\Gamma\left(\frac{d}{2+\alpha}\right) }
-
\frac{\Gamma\left(\frac{d+2}{2+\alpha}\right)^2}{\Gamma\left(\frac{d}{2+\alpha}
\right)^2 } \ .
\end{equation} 
\if
It is useful to recall the formulas
\begin{eqnarray}\label{muvar}
c&=&\int_{\mathbb{R}^d} e^{- |\bw|^{2+\alpha}} 
d\bw=\left|\cS^{d-1}\right|\frac{\Gamma\left(\frac{d}{2+\alpha}\right)
} { 2+\alpha}\nonumber\\
\mu&=&\int_{\mathbb{R}^d} |\bw|^2 
g(\bw)d\bw=\frac{\Gamma\left(\frac{d+2}{2+\alpha}\right)
} {\Gamma\left(\frac{d}{2+\alpha}\right) }\\
\nonumber
\end{eqnarray}
where
\[
 |\mathcal S^{d-1}| = \frac{2 \pi^{\frac{d}{2}}}{\Gamma\left(\frac{d}{2}\right)}
\]
is the surface area of the unit sphere $\mathcal S^{d-1}$. 
\fi
\begin{lem} \label{basic} The following formulas hold for $F_N(\bV)$:
\begin{equation} \label{formulaeff}
F_N(\bV) =  \frac{2+\alpha}{\Gamma(\frac{dN-1}{2 +\alpha})} \frac{1}{Z_N}
\int_0^\infty   t^{dN-1}   \prod_{j=1}^N g( \bv_j t ) \frac{dt}{t} \ ,
\end{equation} 
\begin{equation} \label{formulazee}
 Z_N=\frac{(2+\alpha)}{\Gamma\left(\frac{dN-1}{2+\alpha}\right)}\int_{
\mathbb{R}^{dN}}\frac{
\prod_{i=1}^Ng(\bw_i)}{|\bW|}\,d\bW
\end{equation}
and
\begin{align} \label{formulakay}
f^{(k)}_N(\bV_k) &= \sqrt{\frac{N}{N - |\bV_k|^2}}
\int_{\cS^{d(N-k)-1}\left(\sqrt{N - |\bV_k|^2}\right)}
F_N(\bV_k, \bV^k) d \sigma^{d(N-k)-1}(\bV^k)  \\
&=  \frac{2+\alpha}{\Gamma(\frac{dN-1}{2 +\alpha})Z_N }  \sqrt{\frac{N}{(N -
|\bV_k|^2)^{dk+1}}}\times\nonumber\\
&\phantom{=\frac{2+\alpha}{\Gamma(\frac{dN-1}{2 +\alpha})Z_N
}}\int_{\mathbb{R}^{d(N-k)}}    \prod_{j=1}^k g\left(\frac{\bv_j
|\bW^k|
}{\sqrt{N- |\bV_k|^2}}\right) |\bW^k|^{dk-1}
\prod_{j=k+1}^N g(\bw_j)  d\bW^k \ .\nonumber
\end{align}

\end{lem}
\begin{proof}
Formula \eqref{formulaeff} follows from  \eqref{exN} and 
\begin{equation}\label{exp}
A^{- \gamma} = \frac{1}{\Gamma(\gamma)} \int_0^\infty s^\gamma e^{-As} \frac{ds}{s} \ ,
\end{equation}
valid for all $A>0$ and $\gamma >0$ by setting
\[
A = \left(\sum_{i=1}^N
|\bv_i|^{2+\alpha}\right) \  ,  \ \gamma = \frac{dN-1}{2+\alpha}
\]
and substituting $s =t^{2+\alpha}$. The normalization constant $\widetilde{Z}_N$,
given in \eqref{normalization}, is then 
\[
\widetilde{Z}_N=\frac{2+\alpha}{ \Gamma(\frac{dN-1}{2 +\alpha})} 
\int_{\cS^{Nd-1}(\sqrt N)}
\int_0^\infty   t^{dN-1}  \prod_{j=1}^N
e^{-\left(|\bv_j| t \right)^{2+\alpha}} \frac{dt}{t} d \sigma^{Nd-1} (\bV)
\]
which, using Fubini's theorem and changing variables $\bv_j = \sqrt N \bw_j$
equals 
 \[
 \frac{(2+\alpha)N^{\frac{Nd-1}{2}}}{\Gamma(\frac{dN-1}{2 +\alpha})}
\int_0^\infty  \int_{\cS^{Nd-1}(1)}    \prod_{j=1}^N e^{-( |\bw_j|t\sqrt N)
^{2+\alpha}}  d \sigma^{Nd-1} (\bW)  t^{dN-1}  \frac{dt}{t}  \ . 
  \]
One more variable change $r = \sqrt N t $  yields
\[
 \frac{(2+\alpha)}{\Gamma(\frac{dN-1}{2 +\alpha})}   \int_0^\infty
\int_{\cS^{Nd-1}(1)}    \prod_{j=1}^N e^{-( |\bw_j|r ) ^{2+\alpha}}  d
\sigma^{Nd-1} (\bW)   r^{dN-1} \frac{dr}{r} \ .  
\]
Taking into account the normalization in the definition of $g(w)$ one obtains
\[
Z_N = \frac{(2+\alpha)}{\Gamma(\frac{dN-1}{2 +\alpha})}   \int_0^\infty
\int_{\cS^{Nd-1}(1)}    \prod_{j=1}^N g ( \bw_jr)  d \sigma^{Nd-1} (\bW)
r^{dN-1} \frac{dr}{r} \ ,  
\]
which is the integral  \eqref{formulazee} written in terms of polar 
coordinates. 

To see \eqref{formulakay} we start with \eqref{kmarginal} and find
\[
f_N^{(k)}(\bV_k) =
\sqrt{\frac{N}{N -|\bV_k|^2}}
\frac{2+\alpha}{\Gamma(\frac{dN-1}{2 +\alpha})} \frac{1}{Z_N}
\int_{\cS^{d(N-k)-1}(\sqrt{N -|\bV_k|^2})} \int_0^\infty
t^{dN-1}   \prod_{j=1}^N g (\bv_j t) \frac{dt}{t}  d \sigma^{d(N-k)-1}(\bV^k) \
.
\]
Once more, using Fubini's theorem and changing variables $\bv_j = \sqrt{N -
|\bV_k|^2} \bw_j, j=k+1, \dots, N$ yields
\begin{align}
\sqrt{\frac{N}{N -|\bV_k|^2}}&
\frac{2+\alpha}{\Gamma(\frac{dN-1}{2 +\alpha})} \frac{1}{Z_N}
(N- |\bV_k|^2)^{\frac{d(N-k)-1}{2}} \times\nonumber\\
&\int_0^\infty   t^{dN-1}    \prod_{j=1}^k g(\bv_j t)
\int_{\cS^{d(N-k)-1}(1)}
\prod_{j=k+1}^N g\left(\sqrt{N - |\bV_k|^2} \bw_j t \right) d
\sigma^{d(N-k)-1}(\bW^k) \frac{dt}{t} \ . \nonumber
\end{align}
Changing variables $r = \sqrt{N - |\bV_k|^2} t$ yields the
expression
\begin{align*}
&\sqrt{\frac{N}{N - |\bV_k|^2}}
\frac{2+\alpha}{\Gamma(\frac{dN-1}{2 +\alpha})} \frac{1}{Z_N}
\frac1{(N-  |\bV_k|^2)^{\frac{dk}{2}}} \times\nonumber\\\
&\qquad \int_0^\infty   r^{d(N-k)-1}    \prod_{j=1}^k g\left(\frac{\bv_j
r}{\sqrt{N-
|\bV_k|^2}} \right) r^{dk-1}
\int_{\cS^{d(N-k)-1}(1)}
\prod_{j=k+1}^N g\left(\bw_j r \right)   d \sigma^{d(N-k)-1}(\bW^k) dr
\nonumber \ ,
 \end{align*}
 which equals
 \[
 \frac{2+\alpha}{\Gamma(\frac{dN-1}{2 +\alpha})Z_N}  \sqrt{\frac{N}{(N -
 |\bV_k|^2)^{dk+1}}}
  \int_{\mathbb{R}^{d(N-k)}}    \prod_{j=1}^k g\left(\frac{\bv_j |\bW^k|
}{\sqrt{N-|\bV_k|^2}} \right)|\bW^k|^{dk-1}     \prod_{j=k+1}^N
g\left(\bw_j \right)  d\bW^k \ . 
 \]
\end{proof}
The following elementary lemma will be used to reduce the computation of the 
large $N$ limit of \eqref{formulaeff}, \eqref{formulazee} and \eqref{formulakay} 
to the law of large numbers.
 \begin{lem} \label{pbound} Let $p$ be a probability distribution on $\mathbb R 
^d$ bounded by some constant $C$ and let $0 \le a < d$. Then
\begin{equation} \label{galphabound}
\int_{\mathbb R^d} \frac{p(y)}{|y|^a} dy \le  \frac{d}{d-a} \left( \frac{C 
|\mathcal S^{d-1}|}{d}\right)^{\frac{a}{d}}\ .
\end{equation}
\end{lem}
\begin{proof}
The bathtub principle (see, e.g., \cite{liebloss}) states that the maximum of 
the expression $\int_{\mathbb R^d} \frac{p(y)}{|y|^\alpha} dy$ over all 
probability distributions $p$ with $p(y) \le C$, is attained at  
\[
p^*(y) = 
\begin{cases}
  C \ {\rm if} \ |y| \le R \\ 0 \ {\rm if } \ |y| > R \ ,
\end{cases}
\]
with
\[
R =  \left( \frac{d}{C |\mathcal S^{d-1}|}\right)^{\frac{1}{d}} \ .
\]
The result follows from a straightforward computation.  
\end{proof}
The following serves to demonstrate our simple method with the least amount of fuss.
\begin{lem} \label{example}
Let $a$ be a positive constant and $p$ be a probability distribution bounded by 
$C$ with finite second and fourth moment. Set $m:=\int_{\mathbb R^d} p(y) |y|^2 
dy$. Then
\begin{equation} \label{mennlimit}
\lim_{N \to \infty} \int_{\mathbb R^{Nd}} \left(\frac{\sqrt N}{|\bW|}\right)^a 
\prod_{j=1}^N p(w_i) dw_i =  \left(\frac{1}{\sqrt m}\right)^a \ . 
\end{equation}
\end{lem} 
\begin{proof} We denote  
\[
\mathbb{P}(A):=\int_A \prod_{j=1}^N p(\bw_j) d\bW  \ .
\]
Define the set
\[
A_\varepsilon :=\left\{ \bW \in \mathbb{R}^{Nd} : \left| \frac{|\bW|^2}{N} -
m\right| \le \varepsilon \right\}\ ,
\]
so that
\[
 \left(\frac{1}{\sqrt{ m +\varepsilon} }\right)^a \mathbb{P}(A_\varepsilon) <  
\int_{A_\varepsilon} \left(\frac{\sqrt N}{|\bW|}\right)^a \prod_{j=1}^N p(w_i) 
dw_i
 <  \left(\frac{1}{\sqrt {m - \varepsilon}}\right)^a\mathbb{P}(A_\varepsilon) \ .
\]
Chebyshev's inequality states that
\begin{equation} \label{tcheb}
\mathbb{P}(A^c_\varepsilon) \le \frac{s^2}{\varepsilon^2N} \ ,
\end{equation}
where
\[
s^2 := \int_{\mathbb R^d} p(y) (|y|^2-m^2)^2 dy\ , 
\]
so that
\[
 \left(\frac{1}{\sqrt{ m +\varepsilon} }\right)^a \left(1- 
\frac{s^2}{\varepsilon^2 N}\right) <  \int_{A_\varepsilon} \left(\frac{\sqrt 
N}{|\bW|}\right)^a \prod_{j=1}^N p(w_i) dw_i
 <  \left(\frac{1}{\sqrt {m - \varepsilon}}\right)^a\left(1+ 
\frac{s^2}{\varepsilon^2 N}\right) \ .
\]
It remains to estimate
\[
 \int_{A^c_\varepsilon} \left(\frac{\sqrt N}{|\bW|}\right)^a \prod_{j=1}^N 
p(w_i) dw_i  \ .
\]
By the inequality between the arithmetic and geometric mean
\[
\frac{\sqrt N}{|\bW|} \le \prod_{j=1}^N |w_j|^{-\frac{1}{N}}
\]
and hence
\begin{equation}\label{esti}
\int_{A^c_\varepsilon} \left(\frac{\sqrt N}{|\bW|}\right)^a \prod_{j=1}^N 
p(w_i) dw_i \le
\int_{A^c_\varepsilon}  \prod_{j=1}^N p(w_i) |w_i|^{-\frac{a}{N}}  dw_i =  
\int_{A^c_\varepsilon}  \prod_{j=1}^N \gamma_N(w_i) dw_i
\left(\int_{\mathbb R^d} p(w) |w|^{-\frac{a}{N}}  dw\right)^N
\end{equation}
where
\[
\gamma_N(w) := \frac{ p(w) |w|^{-\frac{a}{N}}}{\int_{\mathbb R^d} p(w) |w|^{-\frac{a}{N}}  dw} 
\]
is a probability measure. It is easy to see that 
\[
m_N := \int_{\mathbb R^d} \gamma_N(w)|w|^2 dw
\]
converges to $m$ as $ N \to \infty$ and hence $m-\varepsilon/2 \le m_N \le m 
+\varepsilon/2$ for all $N$ large enough. Likewise, the fourth moment
\[
s_N^2 : = \int_{\mathbb R^d} \gamma_N(w) (|w|^2 - m^2_N)^2 dw
\]
converges to $\sigma^2$ as $N \to \infty$. Thus we have that the set
\[
B_\varepsilon := \left\{ \bW \in \mathbb{R}^{Nd} : \left| \frac{|\bW|^2}{N} -
m_N \right| \le \frac{\varepsilon}{2} \right\} \subset A_\varepsilon
\]
and hence $A^c_\varepsilon \subset B^c_\varepsilon$ for all $N$ sufficiently large.
Applying Chebyshev's inequality \eqref{tcheb} to the measure $\prod_{j=1}^N 
\gamma(v_j) dv_j$ we find that
\[
 \int_{A^c_\varepsilon}  \prod_{j=1}^N \gamma_N(w_i) dw_i \le  
\int_{B^c_\varepsilon}  \prod_{j=1}^N \gamma_N(w_i) dw_i
 \le \frac{4s_N^2}{\varepsilon^2N}\ .
 \]
Finally, using Lemma \ref{pbound} with $a$ replaced by $\frac{a}{N}$, we get
 \[
 \int_{A^c_\varepsilon} \left(\frac{\sqrt N}{|\bW|}\right)^a \prod_{j=1}^N 
p(w_i) dw_i 
 \le  \frac{4s_N^2}{\varepsilon^2N}  \left(\frac{d}{d-\frac{a}{N}}\right)^N 
\left( \frac{C |\mathcal S^{d-1}|}{d}\right)^{\frac{a}{d}}
 \le  \frac{4s_N^2}{\varepsilon^2N}  \left( \frac{C e  |\mathcal 
S^{d-1}|}{d}\right)^{\frac{a}{d}}
 \]
which tends to zero as $N \to \infty$. Note that we have used the fact that $(1 
- \frac{c}{N})^N$ is monotonically decreasing in $N$. 
\end{proof}
\begin{cor} \label{korollar}
We have the following limit
\[
 \lim_{N \to \infty} \int_{\mathbb{R}^{Nd}} \frac{\sqrt N}{|\bW|} \prod_{j=1}^N
g(\bw_j) d\bW  = \frac{1}{\sqrt \mu} 
\]
so that
\[
\lim_{N \to \infty} \sqrt N \Gamma\left(\frac{dN-1}{2+\alpha}\right)Z_N
=\frac{(2+\alpha)}{\sqrt \mu}
\ . 
\]
\end{cor}

We now turn our attention to $f_N^{(k)}$. According to \eqref{formulakay} we 
have to compute 
\begin{eqnarray}
 \int d \bV_k \varphi(\bV_k) f^{(k)}_N(\bV_k) &=& \nonumber \\
&=& \frac{2+\alpha}{\Gamma(\frac{dN-1}{2 +\alpha})Z_N
}
 \int d \bV_k \varphi(\bV_k)\sqrt{\frac{N}{(N -
|\bV_k|^2)^{dk+1}}}  \nonumber \\
 & \times&\int_{\mathbb{R}^{d(N-k)}}    \prod_{j=1}^k g\left(\frac{\bv_j
|\bW^k|
}{\sqrt{N- |\bV_k|^2}}\right) |\bW^k|^{dk-1}
\prod_{j=k+1}^N g(\bw_j)  d\bW^k \ . \nonumber
\end{eqnarray}
The following change of variables will be helpful.
\begin{lem}\label{change}
Let $\bY=(\bY_k, \bY^k)$ be defined by
\[
\begin{cases}
 \bY_k=\frac{|\bW^k|}{\sqrt{N-|\bV_k|^2}} \bV_k&  \\
\bY^k=\bW^k& \ .
\end{cases}
\]
Then
\begin{equation}
\begin{cases} \label{inverse}
 \bV_k=\frac{\sqrt N}{|\bY|} \bY_k&  \\
\bW^k=\bY^k& \ ,
\end{cases}
\end{equation}
and the Jacobian determinant is given by
\[
 \left|\frac{\partial (\bV_k, \bW^k )}{\partial
(\bY)}\right|= \left(\frac{\sqrt N}{|\bY|} \right)^{dk} \left(\frac{ |\bY^k|^2}{|\bY|^2}\right)
 \ .
\]
\end{lem}
\begin{proof} We have
\[
|\bY_k|^2=\frac{|\bW^k|^2|\bV_k|^2}{N-|\bV_k|^2} = \frac{|\bY^k|^2|\bV_k|^2}{N-|\bV_k|^2}  \ ,
\]
so that
\[
 |\bV_k|^2=\frac{N|\bY^k|^2}{|\bY|^2}
\]
from which \eqref{inverse} follows. 
The Jacobi matrix is of the form
\[
\left[ \begin{array}{cc} A & B \\ 0 & I_{d(N-k)} \end{array}\right]
\]
where $I_n$ is the $n \times n$ identity matrix and the matrices $A,B$ 
are given by
\begin{eqnarray*}
A &=&  \frac{\partial \bV_k}{\partial \bY_k} = 
 \frac{\sqrt N}{|\bY|}
\left(I_{dk} - \frac{\bY_k\otimes \bY_k}{|\bY|^2}\right) \ ,\\
B &=&   \frac{\partial \bV_k}{\partial \bY^k} = 
- \frac{\sqrt N}{|\bY|} \frac{\bY_k \otimes \bY^k}{|\bY|^2}   \  .
\end{eqnarray*}
Note that  $A$ is a $dk 
\times dk$ matrix and $B$ is a $dk \times d(N-k)$ matrix. Hence, the determinant 
 of the Jacobian is given
by $\det A \cdot \det I_{d(N-k)} =  \det A$.  Because
\[
A\bY_k = \frac{\sqrt N |\bY^k|^2}{|\bY|^3}  \bY_k \ ,
\]
$\frac{\sqrt N |\bY^k|^2}{|\bY|^3}$ is a simple eigenvalue and
$ \frac{\sqrt N}{|\bY|}$ is a $(dk-1)$-fold eigenvalue of $A$. We thus find that
\[
\det A =  \left(\frac{\sqrt N}{|\bY|} \right)^{dk} \frac{ |\bY^k|^2}{|\bY|^2}  \ .
\] 
\end{proof}
Let $\varphi(\bV_k)$ be a continuous function on
$\mathbb{R}^{dk}$ such that
\[
 \sup_{\bV_k\in\mathbb{R}^{dk}}\varphi(\bV_k)<K\ .
\]
With the change of variables of Lemma \ref{change} we get 
\begin{eqnarray}\label{inte}
\int d\bV_k\varphi(\bV_k)
f_N^{(k)}(\bV_k)&=&\frac{2+\alpha}{\Gamma(\frac{dN-1}{2
+\alpha})Z_N } \int_{\mathbb{R}^{dk}} d\bY_k\prod_{i=1}^kg(\by_i) \nonumber \\
&\times& \int_{\mathbb{R}^{d(N-k)}}d\bY^k
\frac{\prod_{j=k+1}^N g(\by_j)}{|\bY|}  
\varphi\left(\by_1\frac{\sqrt{N}}{|\bY|} , \ldots
, \by_k\frac {
\sqrt{N}}{|\bY|}\right) \nonumber\\
&=:&\int_{\mathbb{R}^{dk}} d\by_1\cdots d\by_k\prod_{i=1}^kg(\by_i)
H_N(\by_1,\ldots,\by_k) \ .
\end{eqnarray}
with
\begin{equation}
H_N(\bY_k) = \frac{2+\alpha}{\Gamma(\frac{dN-1}{2
+\alpha})Z_N \sqrt N }  \int_{\mathbb{R}^{d(N-k)}}d\bY^k  
\varphi\left(\frac{\sqrt{N}}{|\bY|} \bY_k \right) \prod_{j=k+1}^N g(\by_j)\frac{\sqrt N}{|\bY|} \nonumber \ .
\end{equation}
\begin{lem}  \label{limitnorm}
The function $H_N(\bY_k)$
is bounded on $\mathbb R^{kd}$, in fact
\begin{eqnarray}\label{uniformbound}
|H_N(\bY_k)| &\le& K \left( \frac{N}{N-k}\right)^{\frac{dk+2}{2}} \left( 
\frac{d}{d-\frac{dk+2}{N-k}}\right)^{N-k}\left( \frac{\Vert g \Vert_\infty 
|\mathcal S^{d-1}|}{d}\right)^{\frac{dk+2}{d}} \nonumber \\
&\le& K\left( \frac{N}{N-k}\right)^{\frac{dk+2}{2}} \left( \frac{ e \Vert g 
\Vert_\infty |\mathcal S^{d-1}|}{d}\right)^{\frac{dk+2}{d}} \ ,
\end{eqnarray}
which is bounded uniformly in $N$ for $N> k+1$. Moreover,
\begin{equation} \label{ennlimit}
\lim_{N \to \infty} H_N(\bY_k) =  \left(\frac{1}{\sqrt \mu}\right)^{dk} 
\varphi\left( \frac{\bY_k}{\sqrt \mu}\right) \ .
\end{equation}
\end{lem}
\begin{proof}  Because
\[
\frac{1}{|\bY|} \le \frac{1}{|\bY^k|} 
\]
we get
\[
|H_N(\bY_k)| \le K \frac{2+\alpha}{\Gamma(\frac{dN-1}{2
+\alpha})Z_N }  \int_{\mathbb{R}^{d(N-k)}}d\bY^k  
 \prod_{j=k+1}^N g(\by_j)\frac{1}{|\bY^k|} \le K \frac{Z_{N-k} \Gamma\left( \frac{d(N-k) -1}{2+\alpha}\right)}{Z_N\Gamma(\frac{dN-1}{2
+\alpha})}\ ,
\]
which is bounded uniformly in $N$, in fact the limit as $N \to \infty$ of the 
last expression is $K$. The proof of \eqref{ennlimit} follows, again with a 
slight modification, from the law of large numbers. First, by Corollary 
\ref{korollar} we have
\[
\lim_{N \to \infty}  \frac{2+\alpha}{\Gamma(\frac{dN-1}{2
+\alpha})Z_N \sqrt N }  = \sqrt \mu \ .
\]
We denote  
\[
\mathbb{P}_k(A):=\int_A \prod_{j=k+1}^N g(\bw_j) d\bY^k \ .
\]
Pick $\varepsilon$ small and set 
\[
A_\varepsilon =\left\{ \bY^k \in \mathbb{R}^{(N-k)d} : \left| \frac{|\bY^k|^2}{N-k} -
\mu\right| \le \varepsilon \right\}\ .
\]
Observe that for $\bY^k \in A_\varepsilon$
\[
\sqrt{\frac{N}{N-k}} \frac{1}{\sqrt{\mu+\varepsilon + \frac{|\bY_k|^2}{N-k}}} \le \frac{\sqrt N}{|\bY|} \le
\sqrt{\frac{N}{N-k}} \frac{1}{\sqrt{\mu-\varepsilon + \frac{|\bY_k|^2}{N-k}}}
\]
which, because $\varphi$ is continuous, implies that for $\bY_k$ fixed,
\[
\left| 
\left(\frac{\sqrt N}{|\bY|} \right)\varphi\left(\frac{\sqrt N}{|\bY|} 
\bY_k\right) - \mu^{-\frac{1}{2}}\varphi \left(\frac{\bY_k}{\sqrt 
\mu}\right)\right| = o(\varepsilon)
\]
{\it uniformly} in $\bY^k \in A_\varepsilon$ for $N$ sufficiently large. 
Needless to say this estimate is {\it not} uniform in $\bY_k$, which, however, 
is immaterial for our considerations. From this it follows readily that
\[
\left|\int_{A_\varepsilon}
\frac{\sqrt N}{|\bY|} \varphi\left(\frac{\sqrt N}{|\bY|} 
\bY_k\right)\prod_{j=k+1}^N g(\bw_j) d\bY^k  -\mu^{-\frac{1}{2}} 
\varphi\left(\frac{\bY_k}{\sqrt 
\mu}\right) \mathbb{P}_k(A_\varepsilon)\right| = o(\varepsilon)  \ . 
\]
Using Chebyshev estimate we get
\[
\mathbb{P}_k(A^c_\varepsilon) \le \frac{\sigma^2}{\varepsilon^2 (N-k)} \ ,
\]
so that
\begin{equation} \label{aepsilonestimate}
\left|\int_{A_\varepsilon} 
\frac{\sqrt N}{|\bY|}  \varphi\left( \bY_k  \frac{ \sqrt 
N}{|\bY|}\right)\prod_{j=k+1}^N g(\bw_j) d\bY^k  -\mu^{-\frac{1}{2}} 
\varphi\left(\frac{\bY_k}{\sqrt 
\mu}\right) \right| \le o(\varepsilon) +  \mu^{-\frac{1}{2}} K 
\frac{\sigma^2}{\varepsilon^2 (N-k)} \ .
\end{equation}
It remains to estimate
\begin{eqnarray}
\left | \int_{A^c_\varepsilon} 
\frac{\sqrt N}{|\bY|}  \varphi\left(\bY_k  \frac{ \sqrt 
N}{|\bY|}\right) \prod_{j=k+1}^N g(\bw_j) d\bY^k \right| &\le&  K  
\int_{A^c_\varepsilon}\frac{\sqrt N}{|\bY|}\prod_{j=k+1}^N g(\bw_j) d\bY^k 
\nonumber \\
& \le & K \left(\frac{N}{N-k}\right)^{\frac{1}{2}}  \int_{A^c_\varepsilon} 
\left(\frac{\sqrt{N-k}}{|\bY^k|}\right)\prod_{j=k+1}^N g(\bw_j) d\bY^k \ 
.\nonumber
\end{eqnarray}
Using the same argument as in the proof of Lemma \ref{example} yields the estimate
\[
 \int_{A^c_\varepsilon} \left(\frac{\sqrt {N-k}}{|\bY^k|}\right) \prod_{j=1}^N 
g(y_i) dy_i 
\le  \frac{4\sigma_{N-k}^2}{\varepsilon^2(N-k)}  \left( \frac{C e  |\mathcal S^{d-1}|}{d}\right)^{\frac{1}{d}} \ .
\]
Choosing $\varepsilon = (N-k)^{-\frac{1}{8}}$ and letting $N \to \infty$ proves 
the lemma.
\end{proof}
\begin{proof}[Proof of Theorem \ref{chaoticss}]
\[
\int d \bV_k \varphi(\bV_k) f^{(k)}_N(\bV_k) = \int_{\mathbb{R}^{dk}} 
d\by_1\cdots d\by_k\prod_{i=1}^kg(\by_i)
H_N(\by_1,\ldots,\by_k) \ .
\]
By Lemma \ref{limitnorm} $H_N(\bY_k)$ is bounded uniformly in $N$ and converges 
pointwise to $ \varphi(\bY_k/\sqrt \mu)$ and hence
\[
\lim_{N \to \infty} \int d \bV_k \varphi(\bV_k) f^{(k)}_N(\bV_k) =  \int_{\mathbb{R}^{dk}} d\by_1\cdots d\by_k\prod_{i=1}^kg(\by_i)
\varphi\left(\frac{\bY_k}{\sqrt \mu}\right) \ ,
\]
by the Dominated Convergence Theorem. The last term equals
\[
\int_{\mathbb{R}^{dk}} d\by_1\cdots d\by_k\prod_{i=1}^kf(\by_i) \varphi(\bY_k)
\]
with $f$ given by \eqref{exin}. Note that  $f$ is a probability distribution and 
$\mu$, defined by \eqref{constants}, yields that $\int_{\mathbb R^d} |y|^2 
f(y)dy =1$. This proves the theorem.
\end{proof}

\section{Extension and Remarks.}

It is easy to extend the results of the previous section in a couple of 
interesting directions. We first observe that one can give a stronger definition 
of chaoticity by requiring that given a sequence of normalized functions 
$L_N(\bV)$ on $\cS^{dN-1}(\sqrt{N})$, the entropy per particle of this sequence 
converges to the entropy of the one particle marginal. More precisely if
\[
 S_N=\int_{\cS^{dN-1}(\sqrt{N})}L_N(\bV)\log
L_N(\bV)d\sigma^{dN-1}(\bV)
\]
is the entropy of the $N$ particles system then
\[
 \lim_{N\to \infty}\frac{S_N}{N}=\int_{\mathbb{R}^d} l(\bv)\log
l(\bv)d\bv
\]
where, as before
\[
 l(\bv)=\lim_{N\to\infty} l^{(1)}_N(\bv).
\]
If this is true we say that the sequence $H_N$ is {\bf entropically chaotic},
see \cite{CCLRV}.

\begin{cor}
The sequence $F_N(\bV)$ defined by \eqref{exN} is entropically chaotic and
\begin{equation}\label{SSN}
\lim_{N\to \infty} N^{-1}\int_{\cS^{dN-1}(\sqrt{N})}F_N(\bV)\log
F_N(\bV)d\sigma^{dN-1}(\bV)=\log\left(\frac{\mu^{\frac{d}{2}}}{c}\right)-\frac{d
} {
2+\alpha}=\int_{\mathbb{R}^d} f(\bv)\log f(\bv)d\bv
\end{equation}
where $\mu$ and $c$ are defined in \eqref{constants}.
\end{cor}

\begin{proof}
We will just report here the minor modification to the proof of Lemma
\ref{limitnorm} needed to prove the corollary. We observe that
\[
 x\log x=\lim_{\delta\to 0} \frac{x^{1+\delta}-x}{\delta} \ .
\]
Applying this to \eqref{exp} we get
\[
 A^{-\gamma}\log A^{-\gamma}=\frac{1}{\Gamma(\gamma)}\int_0^\infty
s^{\gamma}\log s^{\gamma}e^{-As}ds-\gamma\psi(\gamma)A^{-\gamma}
\]
where $\psi(x)=\Gamma'(x)/\Gamma(x)$ is the Digamma
function. Following the proof of Lemma \ref{basic}, we find
\begin{eqnarray}\label{SN}
\frac{S_N}{N}&=& -\frac{\log\tilde
Z_N}{N}-\frac{dN-1}{(2+\alpha)N}\psi\left(\frac{dN-1}{2+\alpha}
\right)+\nonumber\\
&&\frac{dN-1}{N}\frac{(2+\alpha)}{\Gamma\left(\frac{dN-1}{2+\alpha}\right)Z_N}
\int_ {\mathbb{R}^{dN}}\log\left(\frac{|\bW|}{\sqrt{N}}\right)\frac{
\prod_{i=1}^Ng(\bw_i)}{|\bW|}\,d\bW
\end{eqnarray}
where $\tilde Z_N=c^{N}Z_N$. Using Stirling formula we get that
\[
 \lim_{N\to\infty}\left(\frac{\log\tilde
Z_N}{N}-\frac{dN-1}{(2+\alpha)N}\psi\left(\frac{dN-1}{2+\alpha}
\right)\right)=-\log c-\frac{d}{2+\alpha}
\]
Finally we need to compute the integral in the last term of \eqref{SN}.
This can be done exactly like in Corollary \ref{korollar} after a simple 
extension of the result in Lemma \ref{example}. Again, we set
$$
A_\varepsilon = \Big\{ \bW : \left| \frac{|\bW|^2}{N} - \mu\right| < \varepsilon \Big \} \ .
$$
For $x <1$ the function $|(\log x)/x|$ is
increasing so that from the inequality of the arithmetic and geometric mean we get
\[
\left|\log\left(\frac{|\bW|}{\sqrt{N}}\right)\frac{\sqrt{N}}{|\bW|}\right|\leq
\frac{\sum_{i=1}^N\left|\log|\bw_i|\right|}{N}\prod
\frac{1}{|\bw_i|^{\frac{1}{N}}}\qquad\mathrm{for\ \ \ 
}\frac{|\bW|}{\sqrt{N}}<1
\]
Proceeding like in the proof of Lemma \ref{example} we only need to modify 
\eqref{esti} as
\begin{align*}
\int_{A^c_\varepsilon, |\bW| <  \sqrt N} \left | \log\left(\frac{|\bW|}{\sqrt{N}}\right) \right| &\frac{\sqrt 
N}{|\bW|} \prod_{j=1}^N 
g(w_i) dw_i \le  \\
&\int_{A^c_\varepsilon}  \prod_{j=1}^N \gamma_N(w_i) dw_i
\left(\int_{\mathbb R^d} g(w) |w|^{-\frac{1}{N}}  dw\right)^{N-1}\int_{\mathbb 
R^d} g(w) \log(|w|)|w|^{-\frac{1}{N}}  dw
\end{align*}
and observe that
\[
\int_{\mathbb 
R^d} g(w) \log(|w|)|w|^{-\frac{1}{N}}  dw<C'
\]
for some constant $C'$ and $N$ large enough. For $x \ge 1$, the non-negative function 
$\frac{\log x}{x}$ is bounded  by $\frac{1}{e}$. Hence
\begin{equation}
\int_{A^c_\varepsilon, |\bW| \ge \sqrt N}  \log\left(\frac{|\bW|}{\sqrt{N}}\right) \frac{\sqrt 
N}{|\bW|} \prod_{j=1}^N 
g(w_i) dw_i \le  
\frac{1}{e} \int_{A^c_\varepsilon}  \prod_{j=1}^N 
g(w_i) dw_i \le  \frac{s^2}{e \varepsilon^2N} 
\end{equation}
using \eqref{tcheb}.
Including the integral over $A_\varepsilon$ and setting $\varepsilon = N^{-1/8}$ yields
\[
 \lim_{N\to\infty}\int_
{\mathbb R^{dN}}\log\left(\frac{|\bW|}{\sqrt{N}}\right)\frac{
\sqrt{N}}{|\bW|}\prod_{i=1}^Ng(\bw_i)\,d\bW=\frac{\log\sqrt{\mu}}{\sqrt{\mu}} \ .
\] 
Combining the above computations with \eqref{SN} we get the first equality in \eqref{SSN}. The
second equality is immediate.

\end{proof}

Another interesting extension is with regards to the first order correction in $\bE$. In
\cite{BCKLs}, under the assumption that the limit $|\bE|\to0$ exists, it was
shown that
\[
 F_{\rm ss}(\bV,\bE)=\delta(U(\bV)-N)\left(F_N(\bV)+
\sum_{i=1}^N \bE\cdot\bc(\omega_i)|\bv_i|R_{N}(\bV)+o(|\bE|)\right) \ ,
\]
where $\bv_i=|\bv_i|\omega_i$ and 
\[
 R_{N}(\bV)=\frac{1}{\widetilde{Z}_N}\frac{dN-1}{\left(\sum_{i=1}^N
|\bv_i|^{2+\alpha}\right)^{\frac{dN-1}{2+\alpha}+1}
}=\frac{1}{|\bv_1|^{2+\alpha}}\bv_1\cdot\nabla_{\bv_1}F_N(\bV) \ .
\]
Here $\bc(\omega)$ is the unique solution of
\[
 \left[(\mathrm{Id}-\cK)\bc\right](\omega)=\omega \ ,
\]
where $\cK$ is the convolution operator generated by $K$, that is
\begin{equation*}
 \left(\cK\bc\right)(\omega)=\int_{\cS^{d-1}(1)}
K(\omega\cdot\omega')\bc(\omega')d\sigma^ { d-1 }
(\omega').
\end{equation*}
Because $-\bc(-\omega)$ is also a solution if $\bc(\omega)$ is, we have, by
uniqueness, that $\bc(\omega)=-\bc(-\omega)$. As a consequence
\[
 \int_{\cS^{d-1}(1)}\bc(\omega')d\sigma^{d-1}(\omega')=0.
\]

Calling $r_{N}^{(k)}$ the marginal of $R_{N}$, we get that
\[
 r_{N}^{(k)}(\bv_1,\ldots,\bv_k)=\frac{1}{|\bv_1|^{2+\alpha}}
\bv_1\cdot\nabla_{\bv_1}f_N^{(k)}(\bv_1,\ldots,\bv_k)
\]
It is easy to see, from \eqref{inte}, that we can take the limit for
$N\to\infty$ on both side and obtain
\[
 \lim_{N\to\infty} r_{N}^{(k)}(\bv_1,\ldots,\bv_k)=r(\bv_1)\prod_{i=2}^k 
f(\bv_k)
\]
where
\[
 r(\bv)=\frac{1}{|\bv_1|^{2+\alpha}}
\bv\cdot\nabla_{\bv}f(\bv)=(2+\alpha)\mu^{\frac{2+\alpha}{2}}f(\bv) \ .
\]
Combining the above results we get that the $k$ particle marginal of $F_{\rm
ss}$ is 
\begin{equation}\label{1oN}
 \lim_{N\to\infty} f_{\rm
ss}^{(k)}(\bv_1,\ldots,\bv_k;\bE)=\left(1+(2+\alpha)\mu^{\frac{2+\alpha}{2}}
\sum_{i=1}^k \bE\cdot\bc(\omega_i)|\bv_i|\right)\prod_{i=1}^k f(\bv_k)+o(|E|)
\end{equation}

This is consistent with the results on the Boltzmann equation \eqref{BE}. To
solve the steady state equation of \eqref{BE} one as to make an assumption on
the form of $\hat\bj_{\rm ss}(E)$ for small $|E|$. It is natural to assume that
\begin{equation}\label{cond}
\hat\bj_{\rm ss}(E)=\tau\underline\kappa E+o(|E|)
\end{equation}
where $\underline\kappa$ is the conductivity tensor for the system with one
particle and energy 1, that is
\[
 \underline\kappa=\frac{1}{|\cS^{d-1}(1)|}\int_{\cS^{d-1}(1)}
\bc(\omega)\otimes\omega\,d\sigma^{d-1}(\omega).
\]
Under this assumption one finds that
\begin{equation}\label{1oinfty}
 f_{\rm
ss}(\bv,\bE)=\left(1+(2+\alpha)\nu^{\frac{2+\alpha}{2}}
\bE\cdot\bc(\omega)|\bv|\right)\tilde f(\bv)+o(|E|)
\end{equation}
where
\[
 \tilde f(\bv)=\frac{\nu^{\frac{d}{2}}}{b}e^{-(\sqrt{\nu} |\bv|)^{2+\alpha}} \ ,
\]
with $\nu$ and $b$ uniquely determined by normalization and \eqref{cond}. One
can also see that the average energy of this solution is
\[
u= \int_{\mathbb{R}^d}
|v^2|\tilde f(\bv)\,dv=\left(\frac{\nu}{\mu}\right)^{\frac{2+\alpha}{2}} \ ,
\]
so that, requiring $u=1$ we obtain once more the large $N$ limit of the one particle
marginal of $F_{\rm ss}$.
Clearly the first order in $E$ of the $k$-fold tensor product of
\eqref{1oinfty} yields \eqref{1oN}. Note that, if the energy per particle is $1$,  the above
results tell us that the current per particle at small field for a large
system is $(2+\alpha)\mu^{\frac{2+\alpha}{2}}$ times the current of the one
particle system.

\subsection*{Acknowledgment} We are indebted to Joel Lebowitz, Ovidiu Costin
and Eric Carlen form many enlightening discussions.  M.L. was supported in part by  NSF grant DMS--0901304.

\bibliographystyle{unsrt}
\bibliography{nonequi}

\end{document}